\documentclass[aps, pra, preprint
]{revtex4-1}
\usepackage{graphicx}
\usepackage{amsfonts}
\usepackage{amssymb}
\usepackage{amsthm}
\usepackage{array}
\usepackage{amsmath}
\usepackage{verbatim} 
\usepackage{hyperref}
\usepackage{color}
\usepackage{bbold}
\usepackage{epstopdf}
\usepackage{mathtools}

\newtheorem{proposition}{Proposition}
\newtheorem{definition}{Definition}
\newtheorem{theorem}{Theorem}

\newtheorem{pcorollary}{Corollary}[proposition]

\newcommand{\ket}[1]{\left\vert#1\right\rangle}
\newcommand{\bra}[1]{\left\langle#1\right\vert}

\newcommand{\abs}[1]{\left|#1\right|}
\newcommand{\ketbra}[2]{\vert #1 \rangle\langle #2 \vert}
\newcommand{\braket}[3]{\langle #1 | #2 |#3 \rangle}
\newcommand{\coh}[1]{\mathcal{C}\left( #1 \right)}
\newcommand{\cc}[1]{\mathcal{C}_{cc}\left( #1 \right)}
\newcommand{\cqc}[1]{\mathcal{C}_{cqc}\left( #1 \right)}
\newcommand{\ecqc}[1]{E_{cqc} \left( #1 \right)}

\def\bra#1{\langle #1|}
\def\ket#1{\left|#1 \right>}

\def\Tr{\mbox{Tr}}

\begin{document}
\title{A Unified View of Quantum Correlations and Quantum Coherence}
\author{Tan Kok Chuan Bobby, Hyukjoon Kwon, Chae-Yeun Park, and Hyunseok Jeong}
\affiliation{Center for Macroscopic Quantum Control, Department of Physics and Astronomy, Seoul National University, Seoul, 151-742, Korea}
\date{\today}

\begin{abstract}
In this paper, we argue that quantum coherence in a bipartite system can be contained either locally or in the correlations between the subsystems. The portion of quantum coherence contained within correlations can be viewed as a kind quantum correlation which we call correlated coherence.
We demonstrate that the framework provided by correlated coherence allows us retrieve the same concepts of quantum correlations as defined by the asymmetric and symmetrized versions of quantum discord as well as quantum entanglement, thus providing a unified view of these correlations. We also prove that correlated coherence can be formulated as an entanglement monotone, thus demonstrating that entanglement may be viewed as a specialized form of coherence.
\end{abstract}

\pacs{}
\maketitle

\section{Introduction}
Quantum mechanics admit a superposition between different physical states. A superposed quantum state is described by a pure state and is completely different in nature to a classical stochastic mixture of states, otherwise called mixed states. In the parlance of quantum mechanics, the former is usually referred to as a coherent superposition, the latter one as a incoherent classical mixture.

A particularly illuminating example of quantum coherence in action is the classic double slit experiment. In the quantum version, a single electron, passing through a double slit one at a time and upon emerging, forms an interference fringe despite not interacting with any other electron. An explanation of this phenomena requires a coherent superposition of two travelling waves emerging from both slits. Such an effect is impossible to explain using only incoherent classical mixtures. Following the birth of quantum theory, physical demonstrations of quantum coherence arising from superpositions of many different quantum systems such as electrons, photons, atoms, mechanical modes, and hybrid systems have been achieved \cite{Hornberger12, Wineland13, Aspelmeyer14}.

Recent developments in our understanding of quantum coherence have come from the burgeoning field of quantum information science. 
One important area of study that quantum information researchers concern themselves with is the understanding of quantum correlations. 
It turns out that in a multipartite setting, quantum mechanical effects allows remote laboratories to collaborate and perform tasks that would otherwise be impossible using classical physics \cite{NielsenChuang}. Historically, the most well studied quantum correlation is quantum entanglement \cite{EPR1935,Horodecki2001, Werner1989}. Subsequent developments of the idea lead to the formulation of quantum discord \cite{Ollivier2001, Henderson2001}, and its symmetrized version \cite{Oppenheim2002, Modi2010, Luo2008} as more generalized forms of quantum correlations that includes quantum entanglement. The development of such ideas of the quantumness of correlations has lead to a plethora of quantum protocols such as quantum cryptography \cite{Ekert91}, quantum teleportation \cite{Bennett1993}, quantum superdense coding \cite{Bennett2001}, quantum random access codes \cite{Chuan2013}, remote state preparation \cite{Dakic2012}, random number generation\cite{Pironio2010}, and quantum computing \cite{Raussendorf2001, Datta2008}, amongst others. Quantum correlations have also proven useful in the study of macroscopic quantum objects \cite{Jeong2015}.

Meanwhile, quantitative theories for entanglement \cite{Plenio07, Vedral98} have been formulated by characterizing and quantifying entanglement as a resource to achieve certain tasks that are otherwise impossible classically.
Building upon this, Baumgratz et al. \cite{Baumgratz14} recently proposed a resource theory of quantum coherence.
Recent developments has since unveiled interesting connections between quantum coherence and correlation, such as their interconversion with each other \cite{Ma15, Streltsov15} and trade-off relations \cite{Xi15}.

In this paper, we demonstrate that quantum correlation can be understood in terms of the coherence contained solely between subsystems.
In contrast to previous studies which established indirect relationships between quantum correlation and coherence \cite{Ma15, Streltsov15, Xi15},
our study establishes a more direct connection between the two and provides a unified view of quantum correlations which includes quantum discord and entanglement using the framework of quantum coherence.

\section{Preliminaries}
\subsection{Bipartite system and local basis}
In this paper, we will frequently refer to a bipartite state which we denote $\rho_{AB}$, where $A$ and $B$ refer to local subsystems held by different laboratories. Following convention, we say the subsystems $A$ and $B$ are held by Alice and Bob respectively.
The local state of Alice is obtained by performing a partial trace on $\rho_{AB}$, and is denoted by $\rho_A = \mathrm{Tr}_B(\rho_{AB})$, and $\{ \ket{i}_A \}$ is a complete local basis of Alice's system.
Bob's local state and local basis are also similarly defined.
In general, the systems Alice and Bob holds may be composite, such that $A=A_1 A_2 \cdots A_N$ and $B = B_1 B_2 \cdots B_M$ so the total state may identically be denoted by $\rho_{A_1A_2\cdots A_N B_1 B_2 \cdots B_M}$.


\subsection{Quantum coherence}
We will adopt the axiomatic approach for coherene measures as shown in Ref.~\cite{Baumgratz14}.
For a fixed basis set $\{ \ket{i} \}$, the set of incoherent states $\cal I$ is the set of quantum states with diagonal density matrices with respect to this basis. Then a reasonable measure of quantum coherence $\mathcal{C}$ should  satisfy following properties:
(C1) $C(\rho) \geq 0$ for any quantum states $\rho$ and equality holds if and only if $\rho \in \cal I$.
(C2a) Non-increasing under incoherent completely positive and trace preserving maps (ICPTP) $\Phi$ , i.e., $C(\rho) \geq C(\Phi(\rho))$.
(C2b) Monotonicity for average coherence under selective outcomes of ICPTP:
$C(\rho) \geq \sum_n p_n C(\rho_n)$, where $\rho_n = \hat{K}_n \rho \hat{K}_n^\dagger/p_n$ and $p_n = \Tr [\hat{K}_n \rho \hat{K}^\dagger_n ]$ for all $\hat{K}_n$ with $\sum_n \hat{K}_n \hat{K}^\dagger_n = \mathbb 1$ and $\hat{K}_n {\cal I} \hat{K}_n^\dagger \subseteq \cal I$.
(C3) Convexity, i.e., $\lambda C(\rho) + (1-\lambda) C(\sigma) \geq C(\lambda \rho + (1-\lambda) \sigma)$, for any density matrix $\rho$ and $\sigma$ with $0\leq \lambda \leq 1$.
In this paper, we will employ the $l_1$-norm of coherence, which is defined by $\coh{\rho} \coloneqq \sum _{i\neq j} \abs{ \bra{i} \rho \ket{j}}$, for any given basis set $\{ \ket{i} \}$ (otherwise called the reference basis). It can be shown that this definition satisfies all the properties mentioned \cite{Baumgratz14}. 

\subsection{Local operations and classical communication (LOCC)}
In addition, we will also reference local operations and classical communication (LOCC) protocols in the context of the resource theory of entanglement. LOCC protocols allow for two different types of operation. First, Alice and Bob are allowed to perform quantum operations, but only locally on their respective subsystems. Second, they are also allowed classical, but otherwise unrestricted communication between them. LOCC operations are especially important in the characterization of quantum entanglement, which typically does not increase under such operations. Measures of entanglement satisfying this are referred to as LOCC monotones \cite{Vidal00}.

\section{Maximal Coherence Loss}

Before establishing the connection between quantum correlation and coherence, we first consider the measurement that leads to the maximal coherence lost in the system of interest. For a monopartite system, the solution to this is trivial. For any quantum state $\rho = \sum_{i,j} \rho_{i,j} \ketbra{i}{j}$ with a reference basis $\{\ket{i}\}$, it is clear that the measurement that maximally removes coherence from the system is the projective measurement $\Pi(\rho) = \sum_{i} \ketbra{i}{i} \rho \ketbra{i}{i}$. This measurement leaves behind only the diagonal terms of $\rho$, so $\coh{\Pi(\rho )} = 0$, which is the minimum coherence any state can have.

A less obvious result for a bipartite state is the following:

\begin{proposition} \label{maxCoh}
For any bipartite state $\rho_{AB} = \sum_{i,j,k,l}\rho_{i,j,k,l} \ket{i,j}_{AB}\bra{k,l}$ where the coherence is measured with respect to the local reference bases $\{\ket{i}_A\}$ and $\{\ket{j}_B\}$, the projective measurement on subsystem $B$ that induces maximal coherence loss is $\Pi_B(\rho_{AB}) = \sum_{j} \left( \openone_A \otimes  \ket{j}_B\bra{j} \right) \rho_{AB}  \left( \openone_A \otimes  \ket{j}_B\bra{j} \right)$.
\end{proposition}

\begin{proof}

We begin by using the spectral decomposition of a general bipartite quantum state $\rho_{AB} = \sum_n p_n \ket{\psi^n}_{AB} \bra{\psi^n}$. Assume that the subsystems have local reference bases $\{\ket{i}_A\}$ and $\{\ket{j}_B\}$ such that $\rho_{AB}= \sum_n \sum_{i,j,k,l} p_n \psi^n_{i,j} (\psi^n_{k,l})^*\ket{i,j}_{AB}\bra{k,l}$. The coherence of the system is measured with respect to these bases. To reduce clutter, we remove the subscripts pertaining to the subsystems $AB$ for the remainder of the proof. Unless otherwise stated, it should be clear from the context which subsystem every operator belong to.

Consider some complete basis on $B$, $\{ \ket{\lambda_m} \}$, and the corresponding projective measurement $\Pi_B(\rho) = \sum_m (\openone\otimes \ketbra{\lambda_m}{\lambda_m}) \,\rho\, (\openone\otimes \ketbra{\lambda_m}{\lambda_m})$. Computing the matrix elements, we get:

\begin{align*}
\braket{i,j}{\Pi_B(\rho)}{k,l} &\coloneqq \left[ \Pi_B(\rho) \right]_{i,j,k,l} \\
&= \sum_n\sum_{p,q} p_n \psi^n_{i,p} (\psi^n_{kq})^*\sum_m \braket{j}{\lambda_m \rangle\langle \lambda_m}{l} \braket{q}{\lambda_m \rangle\langle \lambda_m}{p}.
\end{align*}

Note that minimizing the absolute sum of all the matrix elements will also minimize the coherence, since the diagonal elements of any density matrix always sums to 1 and are non-negative. Consider the absolute sum of all the matrix elements of $\Pi_B(\rho)$:

\begin{align*}
\sum_{i,j,k,l}\abs{\left[ \Pi_B(\rho) \right]_{i,j,k,l}} &= \sum_{i,j,l,k} \abs{\sum_n\sum_{p,q} p_n \psi^n_{i,p} (\psi^n_{k,q})^*\sum_m \braket{j}{\lambda_m \rangle\langle \lambda_m}{l} \braket{q}{\lambda_m \rangle\langle \lambda_m}{p}} \\
&= \left( \sum_{ \substack{i,k \\ j=l}} + \sum_{\substack{i,k \\ j\neq l}} \right) \abs{\sum_n\sum_{p,q} p_n \psi^n_{i,p} (\psi^n_{k,q})^*\sum_m \braket{j}{\lambda_m \rangle\langle \lambda_m}{l} \braket{q}{\lambda_m \rangle\langle \lambda_m}{p}}  \\
&\geq  \sum_{\substack{i,k \\ j=l}} \abs{\sum_n\sum_{p,q} p_n \psi^n_{i,p} (\psi^n_{k,q})^*\sum_m \braket{j}{\lambda_m \rangle\langle \lambda_m}{l} \braket{q}{\lambda_m \rangle\langle \lambda_m}{p}} \\
&= \sum_{i,k,j} \abs{\sum_n\sum_{p,q} p_n \psi^n_{i,p} (\psi^n_{k,q})^*\sum_m \braket{j}{\lambda_m \rangle\langle \lambda_m}{j} \braket{q}{\lambda_m \rangle\langle \lambda_m}{p}} \\
&\geq \sum_{i,k} \abs{\sum_n\sum_{p,q} p_n \psi^n_{i,p} (\psi^n_{k,q})^*\sum_m \sum_j \braket{j}{\lambda_m \rangle\langle \lambda_m}{j} \braket{q}{\lambda_m \rangle\langle \lambda_m}{p}}  \\
&= \sum_{i,k} \abs{\sum_n\sum_{p,q} p_n \psi^n_{i,p} (\psi^n_{k,q})^*\sum_m \braket{q}{\lambda_m \rangle\langle \lambda_m}{p}}  \\ 
&= \sum_{i,k} \abs{\sum_n\sum_{p,q} p_n \psi^n_{i,p} (\psi^n_{k,q})^*\delta_{q,p}}  \displaybreak \\
&= \sum_{i,k} \abs{\sum_n p_n \psi^n_{i,p} (\psi^n_{k,p})^*}
\end{align*}

The first inequality comes from omitting non-negative terms in the sum, while the second inequality comes from moving a summation inside the absolute value function. Note that the final equality is exactly the absolute sum of the elements when $\ket{\lambda_j} = \ket{j}$ since:

$$
\sum_{j} \left( \openone_A \otimes  \ket{j}_B\bra{j} \right) \rho_{AB}  \left( \openone_A \otimes  \ket{j}_B\bra{j} \right) = \sum_{i,j,k} \sum_n p_n \psi^n_{i,j} (\psi^n_{k,j})^* \ketbra{i,j}{k,j}.
$$
This proves the proposition.
\end{proof}
Since any $N$-partite state $\rho_{A_1A_2\ldots A_N}$, is allowed to perform a bipartition such that $\rho_{A_1A_2\ldots A_N} = \rho_{A'A_N}$ where $A' = A_1\ldots A_{N-1}$, we also get the following corollary:

\begin{pcorollary}
For any $N$-partite state $\rho_{A_1A_2\ldots A_N}$ where the coherence is measured with respect to the local reference bases $\{\ket{i}_{A_k}\}$ and $k = 1,2, \ldots, N$, then the projective measurement on subsystem $A_k$ that induces maximal coherence loss is the projective measurement onto the local basis $\{\ket{i}_{A_k}\}$.
\end{pcorollary}

\section{Local and Correlated Coherence}

Now consider a bipartite state $\rho_{AB}$, with total coherence $\coh{\rho_{AB}}$ with respect to local reference bases $\{\ket{i}_A\}$ and $\{\ket{j}_B\}$. Then $\coh{\rho_A}$ can be interpreted as the coherence that is local to $A$. Similarly, $\coh{\rho_B}$ is the portion of the coherence that is local to $B$. In general, the sum of the total local coherences is not necessarily the same as the total coherence in the system. It is therefore reasonable to suppose that a portion of the quantum coherences are not stored locally, but within the correlations of the system itself.

\begin{definition} [Correlated Coherence] With respect to local reference bases $\{\ket{i}_A\}$ and $\{\ket{j}_B\}$, Correlated Coherence for a bipartite quantum system is given by subtracting local coherences from the total coherence: 

$$
\cc{\rho_{AB}} \coloneqq \coh{\rho_{AB}}- \coh{\rho_{A}} - \coh{\rho_{B}}
$$

Where $\rho_{A}$ and $\rho_{B}$ are the reduced density matrices of $A$ and $B$ respectively.
\end{definition}

Further reinforcing the idea that the local coherences form only a portion of the total coherence present in a quantum system, we prove the following property:

\begin{theorem} For any bipartite quantum state $\rho_{AB}$,  $\cc{\rho_{AB}} \geq 0 $ (i.e. Correlated Coherence is always non-negative).
\end{theorem}

\begin{proof}
Let $\rho_{AB}= \sum_n \sum_{i,j,k,l} p_n \psi^n_{i,j} (\psi^n_{k,l})^*\ket{i,j}_{AB}\bra{k,l}$, then

\begin{align*}
\cc{\rho_{AB}} &= \coh{\rho_{AB}}- \coh{\rho_{A}} - \coh{\rho_{B}} \\ 
&= \sum_{\substack{(i,j) \\ \neq (k,l)}} \abs{\sum_n p_n \psi^n_{i,j} (\psi^n_{k,l})^* } - \sum_{i \neq k } \abs{\sum_n p_n\sum_{j} \psi^n_{i,j} (\psi^n_{k,j})^* } -\sum_{j \neq l} \abs{\sum_n p_n \sum_i \psi^n_{i,j} (\psi^n_{i,l})^* } \\ 
&\geq  \sum_{\substack{(i,j) \\ \neq (k,l)}} \abs{\sum_n p_n \psi^n_{i,j} (\psi^n_{k,l})^* } - \sum_{\substack{j \\ i \neq k }} \abs{\sum_n p_n \psi^n_{i,j} (\psi^n_{k,j})^* } -\sum_{\substack{ i \\ j \neq l}} \abs{\sum_n p_n \psi^n_{i,j} (\psi^n_{i,l})^* } \\ 
&= \left( \sum_{\substack{(i,j) \\ \neq (k,l)}} - \sum_{\substack{j = l \\ i \neq k }}-\sum_{\substack{ i = k \\ j \neq l}} \right ) \abs{\sum_n p_n \psi^n_{i,j} (\psi^n_{k,l})^* }.
\end{align*}

The inequality comes from moving a summation outside of the absolute value function. Since $\sum_{\substack{(i,j) \\ \neq (k,l)}}  = \sum_{\substack{j \neq l \\ i \neq k }}+\sum_{\substack{j = l \\ i \neq k }}+\sum_{\substack{ i = k \\ j \neq l}}$, the final equality above is always a sum of non-negative values, which completes the proof.

\end{proof}

\section{Correlated Coherence and Quantum Discord}

Of particular interest to the study of quantum correlations is the idea that certain correlations are quantum and certain correlations are classical. In this section, we will demonstrate that Correlated Coherence is able to unify many of these concepts of quantumness under the same framework.

First, note that in our definition of Correlated Coherence, the choice of reference bases is not unique, while most definitions of quantum correlations are independent of specific basis choices. However, we can retrieve basis independence via a very natural choice of local bases. For every bipartite state $\rho_{AB}$, the reduced density matrices $\rho_A$ and $\rho_B$ has eigenbases $\{ \ket{\alpha_i} \}$ and $\{\ket{\beta_i}\}$ respectively. By choosing these local bases, $\rho_A$ and $\rho_B$ are both diagonal and the local coherences are zero. The implication of this is that for such a choice, \textit{the coherence in the system is stored entirely within the correlations}. Since this can be done for any $\rho_{AB}$, Correlated Coherence with respect to these bases becomes a state dependent property as required. For the rest of the paper, unless otherwise stated, we will assume that the choice of local bases for the calculation of Correlated Coherence will always be the local eigenbases of Alice and Bob.


%

We first consider the definition of a quantum correlation in the symmetrized version of quantum discord. Under the framework of symmetric discord, a state contains quantum correlations when it cannot be expressed in the form $\rho_{AB} = \sum_{i,j} p_{i,j}\ket{i}_A\bra{i}\otimes \ket{j}_B \bra{j}$, where $\{\ket{i}_A\}$ and $\{\ket{j}_B\}$ are sets of orthonormal vectors. Any such state has zero symmetric discord by definition.

We prove the following theorem:

\begin{theorem} [Correlated Coherence and Symmetric Quantum Discord] \label{symDisc}
For a given state $\rho_{AB}$, $\cqc{\rho_{AB}} = 0$ iff $\rho_{AB} = \sum_{i,j} p_{i,j}\ket{i}_A\bra{i}\otimes \ket{j}_B \bra{j}$. 
\end{theorem}

\begin{proof}

If $\{ \ket{i}_A \}$ and $\{\ket{j}_B\}$ are the the eigenbases of $\rho_A$ and $\rho_B$, then $\cqc{\rho_{AB}} = 0$ implies $\coh{\rho_{AB}} = 0$ which implies $\rho_{AB}$ only has diagonal terms, so $\rho_{AB} = \sum_{i,j} p_{i,j}\ket{i}_A\bra{i}\otimes \ket{j}_B \bra{j}$. Therefore, $\cqc{\rho_{AB}} =0 \implies \rho_{AB} = \sum_{i,j} p_{i,j}\ket{i}_A\bra{i}\otimes \ket{j}_B \bra{j}$.

Conversely, if $\rho_{AB} = \sum_{i,j} p_{i,j}\ket{i}_A\bra{i}\otimes \ket{j}_B \bra{j}$, then the state clearly has zero coherence, which implies $\cqc{\rho_{AB}} =0$, so the converse is also true. This proves the theorem.
\end{proof}

This establishes a relationship between Correlated Coherence and symmetric discord. We now consider the asymmetric version of quantum discord. Under this framework, a state contains quantum correlations when it cannot be expressed in the form $\rho_{AB} = \sum_{i} p_{i}\ket{i}_A\bra{i}\otimes \rho_B^i$, where $\rho_B^j$ is some normalized density matrix. 

We prove the following:

\begin{theorem} [Correlated Coherence and Asymmetric Quantum Discord] \label{AsymDisc}
For a given state $\rho_{AB}$, let $\{ \ket{i}_A \}$ and $\{\ket{j}_B\}$ be the the eigenbases of $\rho_A$ and $\rho_B$ respectively. Define the measurement on $A$ onto the local basis as $\Pi_A(\rho_{AB}) \coloneqq \sum_i ( \ket{i}_A\bra{i} \otimes \openone_B) \rho_{AB} ( \ket{i}_A\bra{i} \otimes \openone_B)$. Then, with respect to these local bases, $\cqc{\rho_{AB}} - \cqc{\Pi_A(\rho_{AB})}= 0$ iff $\rho_{AB} = \sum_{i} p_{i}\ket{i}_A\bra{i}\otimes \rho_B^i$, where $\rho_B^i$ is some normalized density matrix and $\{\ket{i}_A\}$ is some set of orthonormal vectors. 
\end{theorem}

\begin{proof}
First, we write the state in the form $\rho_{AB}= \sum_{i,j,k,l}\rho_{ijkl}\ket{i,j}_{AB}\bra{k,l}$. We can always write the state in block matrix form such that $\rho_{AB}= \sum_{i,k}\ket{i}_A\bra{k} \otimes \rho_B^{i,k}$ where $\rho_B^{i,k} \coloneqq \sum_{j,k}\rho_{ijkl}\ket{j}_B \bra{l}$. If $\{ \ket{i}_A \}$ and $\{\ket{j}_B\}$ are the the eigenbases of $\rho_A$ and $\rho_B$, then $\cqc{\rho_{AB}} -  \cqc{\Pi_A(\rho_{AB})}= 0$ implies that when $i \neq k$, $\rho_B^{i,k} = 0$. This implies $\rho_{AB}= \sum_{i}\ket{i}_A\bra{i} \otimes \rho_B^{i,i}$. By defining $\rho_B^i = \rho_B^{i,i}/p_i$ where $p_i \coloneqq \Tr{\rho_B^{i,i}}$, we get $\rho_{AB} = \sum_{i} p_{i}\ket{i}_A\bra{i}\otimes \rho_B^i$. Therefore, $\cqc{\rho_{AB}} - \cqc{\Pi_A(\rho_{AB})}= 0 \implies \rho_{AB} = \sum_{i} p_{i}\ket{i}_A\bra{i}\otimes \rho_B^i$.

For the converse, if $\rho_{AB} = \sum_{i} p_{i}\ket{i}_A\bra{i}\otimes \rho_B^i$, then clearly, $\Pi_A(\rho_{AB})= \rho_{AB}$, so $\cqc{\rho_{AB}} - \cqc{\Pi_A(\rho_{AB})}= 0$. This completes the proof.
\end{proof}

Note that the above relationship with asymmetric quantum discord is expressed as a \textit{difference} between the Correlated Coherence of $\rho_{AB}$ and the post measurement state $\Pi_A(\rho_{AB})$. While this characterization of quantum correlations may at first appear to diverge from the one given in Theorem~\ref{AsymDisc}, they are actually similar since $\cqc{\Pi_A \Pi_B (\rho_{AB})} = 0$ so  $\cqc{\rho_{AB}}= \cqc{\rho_{AB}} - \cqc{\Pi_A \Pi_B (\rho_{AB})}$. It is therefore possible to interpret quantum discord as the correlated coherence loss when either party performs a maximally coherence destroying measurement only their subsystems (See Proposition~\ref{maxCoh}). When the projective measurement is performed only on one side, one retrieves the asymmetric version of quantum discord, and the symmetrized version is obtained when the coherence destroying measurement is performed by both parties.

\section{Correlated Coherence and Entanglement}

 Under the framework of entangled correlations, a state contains quantum correlations when it cannot be expressed as a convex combination of product states $\sum_{i}p_i \ket{\alpha_i}_A \bra{\alpha_i} \otimes  \ket{\beta_i}_B \bra{\beta_i}$, where $\ket{\alpha_i}$ and $\ket{\beta_i}$ are normalized but not necessarily orthogonal vectors that can repeat. It is also possible to extend our methodology to entangled quantum states. In order to do this, we consider extensions of the quantum state $\rho_{AB}$. We say that a state $\rho_{ABC}$ is an extension of $\rho_{AB}$ if $\mathrm{Tr}_C(\rho_{ABC}) = \rho_{AB}$. For our purpose, we will consider extensions of the form $\rho_{AA'BB'}$.

\begin{theorem} \label{separability}
Let  $\rho_{AA'BB'}$ be some extension of a bipartite state $\rho_{AB}$ and choose the local bases to be the eigenbases of $\rho_{AA'}$ and $\rho_{BB'}$ respectively. Then with respect to these local bases, $\min \cqc{\rho_{AA'BB'}} = 0$ iff $\rho_{AB} = \sum_{i}p_i \ket{\alpha_i}_A \bra{\alpha_i} \otimes  \ket{\beta_i}_B \bra{\beta_i}$ for some set of normalized vectors $\ket{\alpha_i}$ and $\ket{\beta_i}$ that are not necessarily orthogonal and may repeat. The minimization is over all possible extensions of $\rho_{AB}$ of the form $\rho_{AA'BB'}$.
\end{theorem}

\begin{proof}
If $\inf \cqc{\rho_{AA'BB'}} = 0$, then $\rho_{AA'BB'}$ must have the form $\sum_{i,j}p_{i,j}\ket{\mu_i}_{AA'}\bra{\mu_i}\otimes \ket{\nu_j}_{BB'}\bra{\nu_j}$ (See Thm.~\ref{symDisc}). Since $\rho_{AA'BB'}$ is an extension, $\mathrm{Tr}_{A'}\mathrm{Tr}_{B'}(\rho_{AA'BB'}) = \sum_{i,j}p_{i,j}\mathrm{Tr}_{A'}(\ket{\mu_i}_{AA'}\bra{\mu_i})\otimes \mathrm{Tr}_{B'} (\ket{\nu_j}_{BB'}\bra{\nu_j}) = \rho_{AB}$. Let $\rho^i_A \coloneqq \mathrm{Tr}_{A'}(\ket{\mu_i}_{AA'}\bra{\mu_i})$ and $\rho^j_B \coloneqq \mathrm{Tr}_{B'}(\ket{\nu_j}_{BB'}\bra{\nu_j})$. Then, $\rho_{AB} = \sum_{i,j}p_{i,j} \rho^i_A\otimes \rho^j_B$. This is equivalent to saying $\rho_{AB}=\sum_{i}p_i \ket{\alpha_i}_A \bra{\alpha_i} \otimes  \ket{\beta_i}_B \bra{\beta_i}$, for some set of (non necessarily orthogonal) vectors $\{\ket{\alpha_i}\}$ and $ \{\ket{\beta_i}\}$. This proves $\min \cqc{\rho_{AA'BB'}} = 0 \implies \rho_{AB} = \sum_{i}p_i \ket{\alpha_i}_A \bra{\alpha_i} \otimes  \ket{\beta_i}_B \bra{\beta_i}$

For the converse, suppose $\rho_{AB} = \sum_{i}p_i \ket{\alpha_i}_A \bra{\alpha_i} \otimes  \ket{\beta_i}_B \bra{\beta_i}$, consider the purification of $\rho_{AB}$ of the form $\ket{\psi}_{ABA'B'C'} = \sum_i \sqrt{p_i}\ket{\alpha_i}_A \ket{\beta_i}_B \ket{i}_{A'}\ket{i}_{B'}\ket{i}_{C'}$. Since this is a purification, $\rho_{AA'BB'} = \mathrm{Tr}_{C'}(\ket{\psi}_{ABA'B'C'}\bra{\psi})$ is clearly an extension of $\rho_{AB}$. Furthermore, the eigenbases of $\rho_{AA'}$ and $\rho_{BB'}$ are $\{ \ket{\alpha_i}_A \ket{i}_{A'}\}$ and $\{ \ket{\beta_i}_B \ket{i}_{B'}\}$ respectively. Since $\rho_{AA'BB'} = \sum_i p_i \ket{\alpha_i}_A \ket{i}_{A'} \bra{\alpha_i}_A \bra{i}_{A'} \otimes \ket{\beta_i}_B \ket{i}_{B'}\bra{\beta_i}_B \bra{i}_{B'}$, $\cqc{\rho_{AA'BB'}} = 0$ with respect to the eigenbases of $\rho_{AA'}$ and $\rho_{BB'}$. Therefore, $\min \cqc{\rho_{AA'BB'}} = 0$, which completes the proof.
\end{proof}

\section{Coherence as an Entanglement Monotone}
We now construct an entanglement monotone using the Correlated Coherence of a quantum state. In order to do this, we first define symmetric extensions of a given quantum state:

\begin{definition} [Unitarily Symmetric Extensions] \label{unitSym}
Let $\rho_{AA'BB'}$ be an extension of a bipartite state $\rho_{AB}$. The extension $\rho_{AA'BB'}$ is said to be unitarily symmetric if it remains invariant up to local unitaries on $AA'$ and $BB'$ under a system swap between Alice and Bob.

More formally, let $\{\ket{i}_{AA'}\}$ and $\{\ket{i}_{BB'}\}$ be complete  local bases on $AA'$ and $BB'$ respectively. Define the swap operator through $U_{\mathrm{swap}} \ket{i,j}_{AA'BB'} \coloneqq \ket{j,i}_{AA'BB'}$ . Then $\rho_{AA'BB'}$ is unitarily symmetric if there exists local unitary operations $U_{AA'}$ and $U_{BB'}$ such that $ U_{AA'}\otimes U_{BB'} \left(U_{\mathrm{swap}} \rho_{AA'BB'} U_{\mathrm{swap}}^\dag \right) U_{AA'}^\dag \otimes U_{BB'}^\dag = \rho_{AA'BB'}$.
\end{definition}

Following from the observation that the minimization of coherence over all extensions is closely related to the the separability of a quantum state, we define the following:

\begin{definition}
Let  $\rho_{AA'BB'}$ be some extension of a bipartite state $\rho_{AB}$ and choose the local bases to be the eigenbases of $\rho_{AA'}$ and $\rho_{BB'}$ respectively. Then the entanglement of coherence (EOC) is defined to be:

$$
E_{cqc}(\rho_{AB}) \coloneqq \min \cqc{\rho_{AA'BB'}}
$$

The minimization is over all possible unitarily symmetric extensions of $\rho_{AB}$ of the form $\rho_{AA'BB'}$.
\end{definition}

It remains to be proven that $E_{cqc}(\rho_{AB})$ is a valid measure of entanglement (i.e. it is an entanglement monotone). But first, we prove the following elementary properties.

\begin{proposition} [EOC of Separable States] \label{separableEcqc}
If a bipartite quantum state $\rho_{AB}$ is separable, $\ecqc{\rho_{AB}} = 0$.
\end{proposition}

\begin{proof}
The proof is identical to Thm.~\ref{separability}, with the additional observation that $\rho_{AA'BB'} = \sum_{i,j}p_{i,j}\ket{\mu_i}_{AA'}\bra{\mu_i}\otimes \ket{\nu_j}_{BB'}\bra{\nu_j}$ is unitarily symmetric. To see this, define $U_{AA'} \coloneqq \sum_{i} \ket{\nu_i}_{AA'}\bra{\mu_i}$ and $U_{BB'} \coloneqq \sum_{i} \ket{\mu_i}_{AA'}\bra{\nu_i}$. It is easy to verify that is satisfies

$$U_{AA'}\otimes U_{BB'} \left(U_{\mathrm{swap}} \rho_{AA'BB'} U_{\mathrm{swap}}^\dag \right) U_{AA'}^\dag \otimes U_{BB'}^\dag = \rho_{AA'BB'}$$

where $U_{\mathrm{swap}}$ is the same as in Def.~\ref{unitSym} so it is unitarily symmetric.
\end{proof}

\begin{proposition} [Invariance under local unitaries] \label{localU}
For a bipartite quantum state $\rho_{AB}$, $E_{cqc}(\rho_{AB})$ is invariant under local unitary operations on $A$ and $B$
\end{proposition}

\begin{proof}
Without loss in generality, we only need to prove it is invariant under local unitary operations of $A$.

For some bipartite state $\rho_{AB}$, let $\rho^*_{AA'BB'} $ be the optimal unitarily symmetric extension such that $ \ecqc{\rho_{AB}}= \cqc {\rho^*_{AA'BB'}}$. Let ${\ket{i}_{AA'}}$ and ${\ket{j}_{BB'}}$ be the eigenbases of $\rho^*_{AA'}$ and $\rho^*_{BB'}$ respectively. With respect to these bases, $\rho^*_{AA'BB'} =  \sum_{ijkl}\rho_{ijkl} \ket{i,j}_{AA'BB'}\bra{k,l}$

Suppose we perform a unitary $U = U_A\otimes \openone_{A'BB'}$ on $A$ such that so $U \ket{i,j} = \ket{\alpha_i, j}$ where $\{\ket{\alpha_i}\}$ is an orthonormal set. Since $U \rho^*_{AA'BB'} U^\dag = \sum_{ijkl}\rho_{ijkl} \ket{\alpha_i ,j}_{AA'BB'}\bra{\alpha_k,l}$, it is clear that the off diagonal matrix elements are invariant under the new bases $\ket{\alpha_i, j}_{AA'BB'}$ so $ \ecqc{\rho_{AB}} = \ecqc {U \rho^*_{AA'BB'} U^\dag }$, which proves the proposition.
\end{proof}

\begin{proposition} [Convexity]
$\ecqc{\rho_{AB}}$ is convex and decreases under mixing:

$$
\lambda \ecqc{\rho_{AB}} + (1-\lambda)\ecqc{\sigma_{AB}} \geq \ecqc{\lambda \rho_{AB} + (1-\lambda) \sigma_{AB}}
$$

For any 2 bipartite quantum states $\rho_{AB}$ and $\sigma_{AB}$, and $\lambda \in [0,1]$.
\end{proposition}

\begin{proof}
Let $\rho^*_{AA'BB'}$ and $\sigma^*_{AA'BB'}$ be the optimal unitarily symmetric extensions for $\rho_{AB}$ and $\sigma_{AB}$ respectively such that $\ecqc{\rho_{AB}}= \cqc {\rho^*_{AA'BB'}}$ and  $\ecqc{\sigma_{AB}}= \cqc {\sigma^*_{AA'BB'}}$.

Consider the state $\tau_{AA'A''BB'B''} \coloneqq \lambda \rho^*_{AA'BB'}\otimes \ket{0,0}_{A''B''}\bra{0,0}+ (1-\lambda) \sigma^*_{AA'BB'}\otimes \ket{1,1}_{A''B''}\bra{1,1}$ for $\lambda \in [0,1]$. Direct computation will verify that with respect to the eigenbases of $\tau_{AA'A''}$ and $\tau_{BB'B''}$, $\cqc{\tau_{AA'A''BB'B''}} = \lambda \cqc{\rho^*_{AA'BB'}} + (1-\lambda) \cqc{\sigma^*_{AA'BB'}} = \lambda \ecqc{\rho_{AB}} + (1-\lambda)\ecqc{\sigma_{AB}}$. However, as $\mathrm{Tr}_{A'A''B'B''}(\tau_{AA'A''BB'B''}) = \lambda \rho_{AB} + (1-\lambda) \sigma_{AB}$, it is an extension of $\lambda \rho^*_{AB} + (1-\lambda) \sigma^*_{AB}$. 

It remains to be proven that the extension above is also unitarily symmetric. Let $\Xi^{\mathrm{swap}}_{ X \leftrightarrow Y}$ denote the swap operation between $X$ and $Y$. Let  the operators $U_{AA'}$, $U_{BB'}$, $V_{AA'}$, $V_{BB'}$ satisfy $\rho^*_{AA'BB'} = U_{AA'}\otimes U_{BB'} \Xi^{\mathrm{swap}}_{ AA' \leftrightarrow BB'}(\rho_{AA'BB'}^*)U_{AA'}^\dag\otimes U_{BB'}^\dag $ and $\sigma^*_{AA'BB'} = V_{AA'}\otimes V_{BB'} \Xi^{\mathrm{swap}}_{ AA' \leftrightarrow BB'}(\sigma_{AA'BB'}^*)V_{AA'}^\dag\otimes V_{BB'}^\dag $ respectively. It can be verified that the local unitary operators $W_{AA'A''} \coloneqq U_{AA'}\otimes \ket{0}_{A''}\bra{0} + V_{AA'}\otimes \ket{1}_{A''}\bra{1}$ and $W_{BB'B''} \coloneqq U_{BB'}\otimes \ket{0}_{B''}\bra{0} + V_{BB'}\otimes \ket{1}_{B''}\bra{1}$ satisfies $\tau^*_{AA'A''BB'B''} = W_{AA'A''}\otimes W_{BB'B''} \Xi^{\mathrm{swap}}_{ AA'A'' \leftrightarrow BB'B''}(\tau_{AA'A''BB'B''})W_{AA'B''}^\dag\otimes W_{BB'B''}^\dag $, so it is also unitarily symmetric.

Since $E_{cqc}$ is a minimization over all unitarily symmetric extensions, we have $ \lambda \ecqc{\rho_{AB}} + (1-\lambda)\ecqc{\sigma_{AB}} = \cqc{\tau_{AA'A''BB'B''}} \geq \ecqc{\lambda \rho_{AB} + (1-\lambda) \sigma_{AB}}$ which completes the proof.
\end{proof}

\begin{proposition} [Contraction under partial trace] \label{partTrace}
Consider the bipartite state $\rho_{AB}$ where $A=A_1A_2$ is a composite system. Then the entanglement of coherence is non-increasing under a partial trace:

$$
\ecqc{\rho_{A_1A_2B}} \geq \ecqc{\mathrm{Tr}_{A_1}(\rho_{A_1A_2B})}
$$

\end{proposition}

\begin{proof}
Let $\rho^*_{A_1A_2A'BB'}$ be the optimal unitarily symmetric extension of $\rho_{A_1A_2B}$ such that $\ecqc{\rho_{A_1A_2B}} = \cqc{\rho^*_{A_1A_2A'BB'}}$. It is clear that $\mathrm{Tr}_{A_1A'B'}(\rho^*_{A_1A_2A'BB'}) = \mathrm{Tr}_{A_1}(\rho_{A_1A_2B}) = \rho_{A_2B}$ so $\rho^*_{A_1A_2A'BB'}$ is an unitarily symmetric extension of $ \mathrm{Tr}_{A_1}(\rho_{A_1A_2B})$. Since $E_{cqc}$ is a minimization over all such extensions, $\ecqc{\rho_{A_1A_2B}} \geq \ecqc{\mathrm{Tr}_{A_1}(\rho_{A_1A_2B})}$.
\end{proof}

\begin{proposition} [Contraction under local projections]

Let $\pi^i_A$ be a complete set of rank 1 projectors on subsystem $A$ such that $\sum_i \pi^i_A = \openone_A$, and define  the local projection $\Pi_A(\rho_A) \coloneqq \sum_i \pi^i_A \rho_A \pi^i_A$ The entanglement of coherence is contractive under a local projections:

$$
\ecqc{\rho_{AB}} \geq \ecqc{\Pi_A(\rho_{AB})}
$$

Or, if $A$ is a composite system, $A=A_1A_2$

$$
\ecqc{\rho_{AB}} \geq \ecqc{\Pi_{A_1}(\rho_{A_1A_2B})}
$$

\end{proposition}

\begin{proof}
First, we observe that any projective measurement can be performed via a CNOT type operation with an ancilla, followed by tracing out the ancilla:

$$\mathrm{Tr}_X\left(U^{\mathrm{CNOT}}_{XY}  \left( \ket{0}_X\bra{0} \otimes \sum_{i,j} \rho_{ij}\ket{i}_Y\bra{j} \right) (U^{\mathrm{CNOT}}_{XY})^\dag \right) = \sum_{i,i} \rho_{ii}\ket{i}_Y\bra{i}.$$

The unitary performs the operation $U^{\mathrm{CNOT}}_{XY} \ket{0,i}_{XY} = \ket{i,i}_{XY}$. Since adding an uncorrelated ancilla does not increase $E_{cqc}$ , we have $\ecqc{\ket{0}_{A_3}\bra{0} \otimes \rho_{A_1A_2B}} = \ecqc{\rho_{A_1A_2B}}$. As $E_{cqc}$ is invariant under local unitaries (Prop.~\ref{localU}) and contractive under partial trace (Prop.~\ref{partTrace}), this proves the proposition.
\end{proof}

\begin{proposition} [Invariance under classical communication] \label{classCom}
For a bipartite state $\rho_{AB}$, Suppose that on Alice's side, $A = A_1A_2$ is a composite system and $A_1$ is a classical registry storing classical information. Then $E_{eqc}$ remains invariant if a copy of $A_1$ is created on Bob's side. 

More formally, let $\rho_{A_1A_2B_2}= \sum_i p_i \ket{i}_{A_1} \bra{i} \otimes \ket{\psi_i}_{A_2B_2}\bra{\psi_i}$ be the initial state, and let $\sigma_{A_1A_2B_1B_2}= \sum_i p_i \ket{i}_{A_1} \bra{i} \otimes \ket{\psi_i}_{A_2B_2}\bra{\psi_i} \otimes \ket{i}_{B_1}\bra{i}$  be the state after Alice communicates a copy of $A_1$ to Bob, then

$$
\ecqc{\rho_{A_1A_2B_2}} = \ecqc{\sigma_{A_1A_2B_1B_2}}
$$
\end{proposition}

\begin{proof}
Let $\Xi^{\mathrm{swap}}_{ X \leftrightarrow Y}$ denote the swap operation between $X$ and $Y$. Let $\rho^*_{A_1A_2A'B_1B_2B'}$ be the optimal unitarily symmetric extension of $\rho_{A_1A_2B_2}$ such that $\ecqc{\rho_{A_1A_2B_2}} = \cqc{\rho^*_{A_1A_2A'B_1B_2B'}}$. Note that $\cqc{\rho^*_{A_1A_2A'B_2B_1B'}} = \cqc{ \ket{0}_{A''}\bra{0} \otimes \rho^*_{A_1A_2A'B_1B_2B'} \otimes \ket{0}_{B''}\bra{0}}$. Define a $\mathrm{CNOT}$ type operation between $A_1$ and $B''$ such that $U^{\mathrm{CNOT}}_{A_1B''} \ket{0,i}_{A_1B''} = \ket{i,i}_{A_1B''}$. Ordinarily, such an operation cannot be done by Bob locally unless he has access to subsystem $A_1$ on Alice's side. However, since $\rho^*_{A_1A_2A'BB'}$ is unitarily symmetric, there exists local unitaries $U_{A_1A_2A'}$ and $U_{BB'}$ such that $\Xi^{\mathrm{swap}}_{A_1A_2A' \leftrightarrow BB'} (\rho^*_{A_1A_2A'BB'}) = U_{A_1A_2A'} \otimes U_{BB'}\rho^*_{A_1A_2A'BB'} U^\dag_{A_1A_2A'} \otimes U^\dag_{BB'}$. This implies that the Bob can perform $U^{\mathrm{CNOT}}_{A_1B''}$ locally by first performing the swap operation through local unitaries, gain access to the information in $A_1$, copy it to $B''$ by performing $U^{\mathrm{CNOT}}_{B_1B''}$ locally, and then undo the swap operation via another set of local unitary operations. 

This means there must exist $V_{A_1A_2A'A''}$ and $V_{B_1B_2B'B''}$ such that:
 $$V_{A_1A_2A'A''} \otimes V_{B_1B_2B'B''} \left( \ket{0}_{A''}\bra{0} \otimes \rho^*_{A_1A_2A'B_1B_2B'} \otimes \ket{0}_{B''}\bra{0} \right) V_{A_1A_2A'A''}^\dag \otimes V_{B_1B_2B'B''}^\dag$$
 is  a unitarily symmetric extension of $U^{\mathrm{CNOT}}_{A_1B''} \left(  \rho_{A_1A_2B_2} \otimes \ket{0}_{B''}\bra{0} \right) U^{\mathrm{CNOT}}_{A_1B''}\, ^\dag$. However, because this state is equivalent to $\sigma_{A_1A_2B_1B_2}$ as defined previously, it is also a unitarily symmetric extension of $\sigma_{A_1A_2B_1B_2}$. Since $\mathrm{C}_{cqc}$ is invariant under local unitary operations, we have

\begin{align*}
 \ecqc{\rho_{A_1A_2B_2}} &= \ecqc{\rho^*_{A_1A_2A'BB'}} \\
 &= \ecqc{\ket{0}_{A''}\bra{0} \otimes \rho^*_{A_1A_2A'BB'} \otimes \ket{0}_{B''}\bra{0}} \\
&= \cqc{V_{A_1A_2A'A''} \otimes V_{BB'B''} \left( \ket{0}_{A''}\bra{0} \otimes \rho^*_{A_1A_2A'BB'} \otimes \ket{0}_{B''}\bra{0} \right) V_{A_1A_2A'A''}^\dag \otimes V_{BB'B''}^\dag} \\
&\geq \ecqc{\sigma_{A_1A_2B_1B_2}},
\end{align*} 
where the last inequality comes from the fact that the entanglement of coherence is a minimization over all unitarily symmetric extensions. On the other hand, a unitarily symmetric extension of $\sigma_{A_1A_2B_1B_2}$ is also an unitarily symmetric extension of $\rho_{A_1A_2B_2}$ so $\ecqc{\rho_{A_1A_2B_2}} \leq  \ecqc{\sigma_{A_1A_2B_1B_2}}$. This implies that $\rho_{A_1A_2B_2}$ so $\ecqc{\rho_{A_1A_2B_2}} =  \ecqc{\sigma_{A_1A_2B_1B_2}}$ which completes the proof.
\end{proof}

\begin{proposition} [Contraction under LOCC] \label{LOCC}
For any bipartite state $\rho_{AB}$ and let $\Lambda_{\mathrm{LOCC}}$ be any LOCC protocol performed between $A$ and $B$. Then $E_{cqc}$ is non-increasing under such operations: 

$$\ecqc{\rho_{AB}} \geq \ecqc{\Lambda_{\mathrm{LOCC}}(\rho_{AB})}$$ 
\end{proposition}

\begin{proof}
We consider the scenario where Alice performs a POVM on her subsystems, communicates classical information of her meaurement outcomes to Bob, who then performs a separate operation on his subsystem based on this measurement information.

Suppose Alice and Bob begins with the state $\rho_{A_1B_1}$. By Naimark's theorem, any POVM can be performed through a unitary interaction between the state of interest and an uncorrelated pure state ancilla, followed by a projective measurement on the ancilla and finally tracing out the ancillary systems. In order to facilitate Alice and Bob's performing of such quantum operations, we add uncorrelated ancillas to the state, which does not change the entanglement of coherence so $\ecqc{\rho_{A_1B_1}}= \ecqc{\ket{0,0}_{M_AA_2}\bra{0,0} \otimes \rho_{A_1B_1}\otimes \ket{0,0}_{M_BB_2}\bra{0,0}}$. For Alice's procedure, we will assume the projection is performed on $M_A$, so $M_A$ is a classical register storing classical measurement outcomes.  

In the beginning, Alice performs a unitary operation on subsystems $M_AA_1A_2$, followed by a projection on $M_A$ which makes it classical. We represent the composite of these two operations with $\Omega_A$, which represents Alice's local operation. Since $E_{cqc}$ is invariant under local unitaries (Prop.~\ref{localU}) but contractive under a projection (Prop.~\ref{partTrace}), $\Omega_A$ is a contractive operation.

The next part of the procedure is a communication of classical bits to Bob. This procedure equivalent to the copying of the state of the classical register $M_A$ to the register $M_B$. However, $E_{cqc}$ is invariant under such communication (Prop.~\ref{classCom}). We represent this operation as $\Gamma_{A\rightarrow B}$.
The next step requires Bob to perform an operation on his quantum system based on the communicated bits. He can achieve this by performing a unitary operation on subsystems $M_BB_1B_2$. We represent this operation with $\Omega_B$, which does not change $E_{cqc}$. The final step of the procedure requires tracing out the ancillas, $\mathrm{Tr}_{M_AA_2M_BB_2}$, which is again contractive (Prop.~\ref{partTrace}). 

Since every step is either contractive or invariant, we have the following inequality:

\begin{align*}
\ecqc{\rho_{A_1B_1}} &= \ecqc{\ket{0,0}_{M_AA_2}\bra{0,0} \otimes \rho_{A_1B_1}\otimes \ket{0,0}_{M_BB_2}\bra{0,0}} \\
&\geq \ecqc{\mathrm{Tr}_{M_AA_2M_BB_2} \circ \Omega_B \circ\Gamma_{A\rightarrow B} \circ \Omega_A \left[\ket{0,0}_{M_AA_2}\bra{0,0} \otimes \rho_{A_1B_1}\otimes \ket{0,0}_{M_BB_2}\bra{0,0}\right]}.
\end{align*}
Any LOCC protocol is a series of such procedures from Alice to Bob or from Bob to Alice, so we must have $\ecqc{\rho_{AB}} \geq \ecqc{\Lambda_{\mathrm{LOCC}}(\rho_{AB})}$, which completes the proof.
\end{proof}

The following theorem shows that the entanglement of coherence is a valid measure of the entanglement of the system.

\begin{theorem} [Entanglement monotone] The entanglement of coherence $E_{cqc}$ is an entanglement monotone in the sense that it satisfies:

\begin{enumerate}
	\item $\ecqc{\rho_{AB}} = 0$ iff $\ecqc{\rho_{AB}}$ is seperable.
	\item $\ecqc{\rho_{AB}}$ is invariant under local unitaries on $A$ and $B$.
	\item $\ecqc{\rho_{AB}}\geq \ecqc{\Lambda_{\mathrm{LOCC}}(\rho_{AB})}$ for any LOCC procedure $\Lambda_{\mathrm{LOCC}}$.
\end{enumerate}
\end{theorem}

\begin{proof}
It follows directly from Prop.~\ref{separableEcqc} and Prop.~\ref{localU} and~\ref{LOCC}.
\end{proof}

\section{conclusion}
To conclude, we defined the Correlated Coherence of quantum states as the total coherence without local coherences, which can be interpreted as the portion of the coherence that is shared between 2 quantum subsystems.
The framework of the Correlated Coherence allows us to identify the same concepts of non-classicality of correlations as those of (both symmetric and asymmetric) quantum discord and quantum entanglement.
Finally, we proved that the minimization of the Correlated Coherence over all symmetric extensions of a quantum state is an entanglement monotone, showing that quantum entanglement may be interpreted as a specialized form of coherence. Our results suggest that quantum correlations in general may be understood from the viewpoint of coherence, thus possibly opening new ways of understanding  both.

\section*{acknowledgment}
This work was supported by the National Research Foundation of Korea (NRF) through a grant funded by the Korea government (MSIP) (Grant No. 2010-0018295).

\end{document}